\documentclass[11pt,a4paper, final, twoside]{article}
\usepackage[utf8]{inputenc}
\usepackage[T1]{fontenc}

\usepackage{amsmath}
\usepackage{amsthm}
\usepackage{amsfonts}
\usepackage{amssymb}
\usepackage{textcomp}
\usepackage{amscd}
\usepackage{latexsym}
\usepackage{amsthm}
\usepackage{graphicx}
\usepackage{graphics}
\usepackage{afterpage}
\usepackage[colorlinks=true, urlcolor=blue,  linkcolor=black, citecolor=red]{hyperref}
\usepackage{color}

\newtheorem{thm}{Theorem}[section]

\newtheorem{example}[thm]{Example}

\newtheorem{notation}[thm]{Notation}

\theoremstyle{proposition}
\newtheorem{prop}{Proposition}[section]

\theoremstyle{definition}
\newtheorem{defn}{Definition}[section]
\theoremstyle{remark}
\newtheorem{rem}{Remark}[section]

\numberwithin{equation}{section}
\begin{document} 

\title{ERROR CORRECTING CODES FROM SUB-EXCEEDING FUNCTIONS}
\author{Luc RABEFIHAVANANA $^1$ \\Harinaivo ANDRIATAHINY $^2$ \\Toussaint Joseph RABEHERIMANANA $^3$\\[6pt]
 $^{1,2,3}$ Department of Mathematics and Computer Science,\\Faculty of Sciences, University of Antananarivo\\[6pt]
 $^{1}$ e-mail: lucrabefihavanana@yahoo.fr,\\$^{2}$e-mail: aharinaivo@yahoo.fr\\$^{3}$ e-mail: rabeherimanana.toussaint@yahoo.fr}
\maketitle
\begin{abstract}
In this paper, we present error-correcting codes which are the results of our research on the sub-exceeding functions. For a short and medium distance data transmission (wifi network, bluetooth, cable, ...), we see that these codes mentioned above present many advantages compared with the Hamming code which is a 1 correcting code.
\medskip

{\bf Math. Subject Classification:} 94B05; 94B35; 11B34

{\bf Key Words and Phrases:} Error correction code, encoding, decoding, sub-exceeding function 

\end{abstract}

\section{Introduction}
New information and communication technologies or NICTs  require today a norm increasingly strict in terms of quality of service. The diversity and the increasing volumes of data exchanged/processed also require increasingly fast and reliable systems.

In these constraints related to information processing, we  need to take into account the increased sensitivity of technologies in front of external disruptive sources. It's about especially to protect information against environmental damage during transmission.\\
In seeking to improve this quality of service in terms of communication, the present article lists the results which can be increase the reliability in the data transmission system. These are the fruits of the in-depth study on our two articles entitled: "Part of a set and sub-exceeding function (encoding and Decoding) \cite{Luc1}" in 2017 and "Encoding of Partition Set Using Sub-out Function \cite{Luc2} "in 2018. \\
Given a finite-dimensional set $ \Omega $ , the first article discusses the encoding (decoding) of any subset of $ \Omega $. For the second, it presents the encoding (decoding) of any partition of the set $\Omega $. \\
The purpose of this article is to build new error correcting code using the results in the two above-mentioned articles.
\\
Given an integer $ k $ such that $ k \geq 3$, we can build two error-correcting codes $\mathcal{L}_k$ and $\mathcal{L}_k ^{+}$ which are respectively $[2k,k]$ and $[3k,k]$ linear codes.\\
In the last section of this article, we give the decoding algorithm of these codes using Groebner basis.

\section{Preliminaries}
Let $n$ be a positive integer, $ [n] $ denotes the set of positive integers less or equal to  $ n $ and the zero element, i.e. $$ [n] = \left \lbrace 0, 1, \; 2, \; ... \; n \right \rbrace. $$
\subsection{The necessary ones on the study of sub-exceeding functions}
\begin{defn}(see \cite{Luc1})
Let $ n $ be a  positive integer and let $ f $ be a map from $ [n] $ to $ [n] $. This function $ f $  is said sub-exceeding if for all $ i $ in $ [n] $, we have $$ f (i) \leq i. $$
\end{defn}
We denote by $ \mathcal{F}_{n} $ the set of  all sub-exceeding functions on $ [n] $, i.e.
 \begin{equation}
 \mathcal{F}_{n}=\left\lbrace f:\;[n]\;\longrightarrow\;[n]\;\mid\;f(i)\leq i,\;\forall i\in [n]\right\rbrace.  \end{equation}
 
\begin{rem}
A sub-exceeding function $ f $ can be represented by the word of $ n + 1 $ alphabet $ f (0) f (1) f (2) ... f (n) $. So, we describe $ f $ by its images $ f = f (0) f (1) f (2) ... f (n) $. 
\end{rem}
\begin{defn}\label{def1}(see \cite{Luc1})
Let $ n $ and $ k $ be two integers such that $ 0 \leq k \leq n $. We define by $ \mathcal {H} ^ {k} _ {n} $ the subset of $ \mathcal {F} _ {n} $ such that
\begin{equation}
\mathcal{H}^{k}_{n}=\left\lbrace f\in\mathcal{F}_{n} \hspace{0.5 cm}|\hspace{0.5 cm}f(i)\leqslant f(i+1) \hspace{0.5 cm} \forall\; 0\leqslant i\leqslant n \hspace{0.5 cm} \text{and} \hspace{0.2 cm} Im(f)=[k] \right\rbrace \;. 
\end{equation} 
\end{defn}
Here, $\mathcal{H}^{k}_{n}$ is the set of all sub-exceeding function of $\mathcal{F}_{n}$ with a quasi-increasing sequence of image formed by all elements of  $[k]$. 
 
\begin{example}
Take $ n = $ 4 and $ k = $ 3. We find that the function $ f = 01123 $ is really in $ \mathcal {H} ^ {3}_{4} $ because $ (f (i)) _ {0\leq i\leq 4} $ is a quasi-increasing sequence formed by all the elements of $ [3] $. But if we take $ f = 01133 $, even if the sequence $ (f (i)) _ {0\leq i\leq 4} $ is quasi-increasing, $ f = 01133 \notin \mathcal {H} ^ {3 } _ {4} $ because $ Im (f) \neq [3] $ (without 2 among the $f(i)$).

\end{example}
Following the definition \ref {def1}, we denote by $ \mathcal {H} _ {n} $ the set defined as follows:
\begin{equation}
\mathcal{H}_{n}=\bigcup^{n}_{k=0}\mathcal{H}^{k}_{n}.
\end{equation}
\begin{thm}\label{Hn}( See \cite{Luc1})\\
Let $ n $ and $ k $ be two integers such that $ 0 \leq k \leq n $.

\begin{enumerate}
\item For $ k = 0 $, we always find that $ \mathcal {H} ^ {0} _ {n} $ is a set of singleton:
\begin{equation*}
\mathcal{H}^{0}_{n}\;=\;\left\lbrace f\;=\;000...00_{\text{n+1-terms}}\right\rbrace .
\end{equation*}
\item For $ k = n $, we also find that $ \mathcal {H} ^ {n} _ {n} $ is a set of singleton:
\begin{equation*}
\mathcal{H}^{n}_{n}\;=\;\left\lbrace f\;=\;0123....(n-1)(n)\right\rbrace .
\end{equation*}
\item For any integer $ k $ such that $ 0 <k <n $, we can construct all sub-exceeding functions of $ \mathcal {H} ^ {k} _ {n} $ as follows:
\begin{enumerate}
\item Take all the elements of $ \mathcal {H} ^ {k-1} _ {n-1} $ and add the integer $ k $ at the end,
\item Take all the elements of $ \mathcal {H} ^ {k} _ {n-1} $ and add the integer $ k $ at the end
\end{enumerate}
To better presentation of this construction, we adopt the following writing:
\begin{equation*}
\mathcal{H}^{k}_{n}\;=\;\left\lbrace \mathcal{H}^{k-1}_{n-1}\curvearrowleft k\right\rbrace \;\bigcup\;\left\lbrace \mathcal{H}^{k}_{n-1}\curvearrowleft k\right\rbrace .
\end{equation*}
Here, $ (*) \curvearrowleft k $ means that we add the integer $ k $ at the end of all element of $ (*) $.
\end{enumerate}
\end{thm} 
From this theorem, we have the iteration table of the elements of $ \mathcal {H} ^ {k} _ {n} $:

\begin{center}
\begin{tabular}{|c|c|c|c|c|c|c}
\cline{1-7}
$n\hspace{0.5 cm}\smallsetminus\hspace{0.5 cm}k$ & 0 & 1 & 2 & 3 & 4 & ...\\\hline
0 & \color{green}{0}  & & & & & \\\hline
1 & 0\color{green}{0} & 0\color{green}{1} & & & & \\\hline

 & & 00\color{green}{1} &  & & & \\
2 & 00\color{green}{0} & 01\color{green}{1} & 01\color{green}{2} & & & \\\hline
 
 &  & 000\color{green}{1} & 001\color{green}{2} & & & \\
3 & 000\color{green}{0} & 001\color{green}{1} & 011\color{green}{2} & 012\color{green}{3} & & \\
 &  & 011\color{green}{1} & 012\color{green}{2} &  & & \\\hline

 &  & 0000\color{green}{1} & 0001\color{green}{2} & 0012\color{green}{3} & & \\
 &  & 0001\color{green}{1} & 0011\color{green}{2} & 0112\color{green}{3} &  & \\
 & 0000\color{green}{0} & 0011\color{green}{1} & 0111\color{green}{2} & 0122\color{green}{3} & 0123\color{green}{4} & \\
4 &  & 0111\color{green}{1} & 0012\color{green}{2} & 0123\color{green}{3} &  & \\
 &  &  & 0112\color{green}{2} &  &  & \\
 &  &  & 0122\color{green}{2} &  &  & \\\hline
\vdots &  & & & & & \\
\end{tabular}
\end{center}

\begin{prop}\label{prop 2.1} See \cite{Luc1}\\
Let $n$ and $k$ be two integers such that $0\leqslant k\leqslant n$. So, we have the following relations
\begin{enumerate}
\item  $ \text{Card}\hspace{0.2 cm} \mathcal{H}^{k}_{n}= \text{Card}\hspace{0.2 cm} \mathcal{H}^{k-1}_{n-1} + \text{Card}\hspace{0.2 cm} \mathcal{H}^{k}_{n-1}$
\item $ \text{Card}\hspace{0.2 cm} \mathcal{H}^{k}_{n}= \left( \begin{matrix}
n\\
k
\end{matrix}\right) $ 
\item $ \text{Card}\hspace{0.2 cm} \mathcal{H}_{n}= 2^{n}$
\end{enumerate}
\end{prop}
From this Proposition, we have the cardinal table of  $\mathcal{H}^{k}_{n}$ as
\begin{center}
\begin{tabular}{|c|c|c|c|c|c|c}
\cline{1-7}
$n\hspace{0.5 cm}\smallsetminus\hspace{0.5 cm}k$ & 0 & 1 & 2 & 3 & 4 & ...\\\hline
0 & \color{green}{1}  & & & & & \\\hline
1 & \color{green}{1} & \color{green}{1} & & & & \\\hline
2 & \color{green}{1} & 2 & \color{green}{1} & & & \\\hline
3 & \color{green}{1} & 3 & 3 & \color{green}{1} & & \\\hline
4 & \color{green}{1} & 4 & 6 & 4 & \color{green}{1} & \\\hline
\vdots &  & & & & & \\
\end{tabular}
\end{center}
Thus constructed, this table is none other than the Pascal triangle.
\subsection{Some notion about error-correcting codes}
When sending a message (data) through a transmission channel (by downloading data from Internet for example), errors can occur. Our goal is to be able to detect or correct these errors.
The principle of coding is as follows: after cutting our message into blocks of $ k $ bits, we will apply the same algorithm on each block:
\begin{itemize}
\item[-] by adding check bits at the end of all block
\item[-] or by completely modifying the blocks, but avoiding that two different blocks are transformed into a single block.
\end{itemize}
\begin{notation}
A block sequence of $ k $ bits will be indifferently called word or vector and will be denoted $ \;\;m_ {1} m_ {2} ... m_ {k} $.
\end{notation}

\begin{defn} {[A numerical code]}(see \cite{PA2})
\begin{itemize}
\item[-] A binary linear code is the image of an injective linear application $ \phi$ from $\left \lbrace 0,1 \right \rbrace ^ {k}$ to $\left \lbrace 0,1 \right \rbrace ^ {n} $. The parameter $ k $ is called the \textbf {dimension} and $ n $ is called the \textbf {length} of this code. So, we say that $ \phi $ is a code of parameters $[ n, k ]$  or $[n,k]$-linear code.
\item[-] The set $ \mathcal {C} = \left \lbrace \phi (m) \mid\; m \in \left \lbrace 0,1 \right \rbrace ^ {k} \right \rbrace $ is called the image of  $ \phi $ and the elements of $ \mathcal {C} $ are called the codewords(as opposed to the original elements of $ \left \lbrace 0,1 \right \rbrace ^ {k} $ which are called messages).
\end{itemize}
\end{defn}

\begin{rem}
A binary linear code $ \mathcal {C} $ of parameters $ [n, k] $ or $ [n,k] $ linear code is a subspace of $ \mathbb {F} _ {2} ^ {n} $ of dimension $k$.
\end{rem}
\begin{defn}(see \cite{PA2})
\begin{itemize}
\item[-]  The minimum distance $ d $ of a code $ \mathcal {C} $  is defined as,
\begin{equation}
d\;=\;min\left\lbrace d(x,y)\mid x\in\mathcal{C}; y\in \mathcal{C}; x\neq y\right\rbrace .
\end{equation}
Here $ d (x, y) $ presents the Hamming distance between $ x $ and $ y $.
\item[-]  A linear code with length $n$, dimension $k$ and minimum distance $d$ is called a $[n,k,d]$-linear code.
\end{itemize}
\end{defn}

\begin{defn}(see \cite{PA2})
\begin{itemize}
\item[-] We define by error detection capacity for the  code $ \mathcal {C} $, the maximum number of errors which we can be detected for a wrong code word. This integer is denoted by $ e_ {d} $.
\item[-] We call by correction capacity of the code $ \mathcal {C} $ the maximum number of errors which we can be corrected  for a wrong code word. This integer is denoted by $ e_ {c} $
\end{itemize}
\end{defn}
\begin{thm}(see \cite{PA1})\;
For an $[n, k, d]$-linear code, that is to say a code of length $ n $ and of dimension $ k $ with a minimal distance $ d $, we have:
\begin{equation}
e_{d}=d-1 \,\,\text{ and}\;\; e_{c}=\lfloor\dfrac{d-1}{2}\rfloor
\end{equation}
\end{thm}
\begin{thm}(see \cite{PA1})
Let $ \mathcal {C} $ be a linear code, we have
\begin{equation*}
d\;=\;min\left\lbrace w(c)\mid c\in\mathcal{C} \text{ with } c\neq 0\right\rbrace .
\end{equation*}
Here $ w (c) $ is the weight of the  codeword $ c $, i.e. the number of non-zero bits of $ c $.
\end{thm}
\begin{thm} (see \cite{PA1})
Let $ \mathcal {C} $ be a $[n, k, d]$-linear code. There is a linear application $ \phi $ from $ \mathbb {F} _ {2} ^ {k} $ to $ \mathbb {F} _ {2} ^ {n} $ such that the image of $ \phi $ is $ \mathcal {C} $. So, we can represent the application $ \phi $ by $ \phi (m) = m \times G $ where $ G $ is the representative matrix of $ \phi $ by the canonical bases of $ \mathbb {F} _ {2} ^ {k} $. \\
Here $G$ is called the generating matrix of $ \mathcal {C} $.
\end{thm}
\begin{defn}(see \cite{PA1})
A  $[n, k, d]$-linear code $C$ is called systematic if the encoding consists in adding $ n-k $ bits at the end of the message. For a linear code, The generating matrix $ G $  is  of the form
$\left(\begin{array}{cc}
I_{k} & G'\\
\end{array}\right) $, where $ I_{k} $ is the unit matrix with $ k $ rows and $ k $ columns and $ G '$ a matrix with $ k $ rows and $ n-k $ columns.
\end{defn}
\begin{defn}(see \cite{PA2})
A parity check matrix for a binary $[n, k, d]$-linear code $C$ is the matrix $ H \in \mathcal {M} _ {n, n-k}(\mathbb{F}_{2}) $ such that $H\times\;^{t}c=0$ if and only if $c\in C$.
\end{defn}
\section{Main result: \\ error-correcting code from the study on the sub-exceeding function }
In this section, we present our linear error-correcting code from sub-exceeding function.\\

Recall that for a positive integer $ n $, a function $ f $  from $ [n] $ to $ [n] $ is said to be sub-exceeding if for any integer $ i $ in $ [n] $, we always have the inequality $ f (i) \leq i $.

Thus, the sub-exceeding term amounts to saying that the image of an integer $ i $ by an application $ f $ is always an integer smaller or equal to this one.

\begin{thm}
Let $ k $ be a positive integer and let $ f $ be an application from $[k] $ to $ \mathbb {F} ^ {k+1} _ {2} $. Then the application $ f $  is a sub-exceeding function if and only if $f(0)=0$.
\end{thm}
This theorem tells us that all message of $ k $ bits on $ \mathbb {F} _ {2} $ which begins with 0 is a sub-exceeding function.
\begin{proof}
We say that the image of an integer $i$ in $[k]$ by the application $ f $ is always equal to 0 or 1. Thus, by the condition $f(0)=0$, we have $ f (i) \leq i $ for all $i$. So, $f$ is a sub-exceeding function.
\end{proof}
Now, let's examine the subset $ \mathcal {H} _ {k} $ for the set of sub-exceeding functions in all application from $ [k] $ in $ \mathbb {F} ^ {k+1} _ {2} $. That is to say the subset $ \mathcal {H} _ {k} $ for the set of $ k $ bits messages on $ \mathbb {F} _ {2} $. \\
Referring to the theorem (\ref {Hn}), we can have all the elements of $ \mathcal {H} _ {k} $ (see the table below). Moreover, from the proposition (\ref {prop 2.1}), we find

\begin{equation}
\text{Card } (\mathcal{H}^{0}_{k}) = 1 \text{ and that  Card}(\mathcal{H}^{1}_{k}) = k .
\end{equation}
The table below shows the elements of $ \mathcal {H} ^ {i} _ {k} $ for each value of $ i \in \left \lbrace 0,1 \right \rbrace $ and some integer $ k $.
\begin{equation*}
\begin{array}{r|r|r}
 n\;\;\backslash\;\;k& 0 & 1\\\hline
  0 & 0 & \\\hline
  1 & 00 & 01\\\hline
  2 & 000 & 001 \\
  & & 011 \\\hline
  & & 0001\\  
  3 & 0000 & 0011 \\    
  & & 0111\\\hline
   & & 00001\\
 4  & 00000 & 00011\\
   & & 00111\\
   & & 01111\\\hline
      &  & 000001\\
      & & 000011\\
    5  & 000000& 000111 \\
      & & 001111\\
      & & 011111\\\hline 
     \vdots  & & \\   
\end{array}
\end{equation*}
\begin{prop}
Let $ k $ be a positive integer and denote by $ T_ {k} $ the matrix of $ k + 1 $ row and $ k $ column such that 
\begin{equation}
T_{k}=\left( \begin{array}{ccccccc}
0 & 0 & 0 & ...& 0  & 0 & 1 \\
0 & 0 & 0 & ...& 0  & 1 & 1 \\
0 & 0 & 0 & ...& 1  & 1 & 0 \\
\vdots & \vdots  & \vdots & \diagup& \diagup  & \vdots & \vdots \\
0 & 0 & 1  & 1 & 0 & ... & 0 \\
0 & 1 & 1 & 0 & 0  & ... & 0 \\
1 & 1 & 0 & 0 & 0  & ... & 0 \\
0 & 1 & 1 & 1 & 1  & ... & 1 \\

\end{array}\right). 
\end{equation}
Reminding that $ \mathcal {H} ^ {1} _ {k} $ is the set of sub-exceeding functions $ f_i $ of length $ k + 1 $ such that
\begin{equation*}
\begin{array}{ccc}
f_{i}=000 & ... & 0\underbrace{11...1} \\
 &  &\;\; i-times
\end{array}, \text{ with } i\in \left\lbrace 1,... k\right\rbrace .
\end{equation*}
Then, the product $ f_ {i} \times T_ {k} $ gives the word $ g_ {i} $ such that

\begin{equation}\label{Gk+1}
\left\lbrace  \begin{array}{cc}
 g_{1} = & 0111...111 \\
 g_{2} = & 1011...111  \\
 g_{3} = & 1101...111 \\
 \vdots & \vdots \\
 g_{k-1} = & 1111...101 \\
 g_{k} = & 1111...110 \\
 
 \end{array}\right\rbrace .
\end{equation}
\end{prop}
\begin{notation}
Now let's denote by $ G_ {k} $ the matrix
\begin{equation}
G_{k}=\left(   \begin{array}{c}
 g_{1} \\
 g_{2}  \\
 g_{3}  \\
 \vdots \\
 g_{k} \\
  \end{array}\right)  =  \left(   \begin{array}{cc}
   0111...111 \\
   1011...111  \\
   1101...111 \\
    \vdots \\
   1111...101 \\
   1111...110 \\
      \end{array}\right).
 \end{equation}
\end{notation}

\begin{example}
For $ k = $ 3, we have:
\begin{center}
$G_{3}\left\lbrace  \begin{array}{cc}
 g_{1} = & 011 \\
 g_{2} = & 101  \\
 g_{3} = & 110 \\
 \end{array}\right\rbrace, T_{3}=\left( \begin{array}{ccc}
0 & 0 & 1 \\
0 & 1 & 1 \\
1 & 1 & 0 \\
0 & 1 & 1
\end{array}\right) \text{ where } \mathcal{H}^{1}_{3}\left\lbrace  \begin{array}{cc}
 f_{1} = & 0001 \\
 f_{2} = & 0011  \\
 f_{3} = & 0111 \\
 \end{array}\right\rbrace $
\end{center}
For $k=4$, we have:

\begin{center}
$G_{4}\left\lbrace  \begin{array}{cc}
 g_{1} = & 0111 \\
 g_{2} = & 1011  \\
 g_{3} = & 1101 \\
 g_{4} = & 1110 \\
 \end{array}\right\rbrace, T_{4}=\left( \begin{array}{cccc}
0 & 0 & 0 & 1  \\
0 & 0 & 1 & 1  \\
0 & 1 & 1 & 0  \\
1 & 1 & 0 & 0  \\
0 & 1 & 1 & 1
\end{array}\right) \text{ where } \mathcal{H}^{1}_{4}\left\lbrace  \begin{array}{cc}
 f_{1} = & 00001 \\
 f_{2} = & 00011 \\
 f_{3} = & 00111 \\
 f_{4} = & 01111 \\
 \end{array}\right\rbrace $
\end{center}
\end{example}
\begin{thm}\label{thmpric1}\textbf{(Main theorem: The codes $\mathcal{L}_k$)}\\
Let $ k $ be a positive integer and let $ \psi $ be the linear application from $ \mathbb {F} _ {2} ^ {k} $ to $ \mathbb {F} _ {2} ^ {2k} $ such that
\begin{equation}
\begin{array}{cccc}
\psi\;: & \mathbb{F}_{2} ^{k} & \longrightarrow & \mathbb{F}_{2} ^{2k} \\
 & m & \longmapsto & \psi(m)= m\;\times\;G
\end{array} 
\end{equation}
where $ m $ is the message of $ k $ bits such that $ m = m_ {1} m_ {2} ... m_ {k} $ and $ G $ is the matrix
$ G =\left(\begin{array}{cc}
I_{k} & G_{k}\\
\end{array}\right)$ where $ G_ {k} $ is the matrix defined in the equation (\ref{Gk+1}). \\
Thus, the application $ \psi $  forms a systematic $[2k,k]$-linear error-correcting code denoted by $\mathcal{L}_k$ . The minimum distance of $\mathcal{L}_3$ is 3 and for $k\geq 4$ the minimum distance of  $\mathcal{L}_k$ is 4. 
\end{thm}
\begin{proof}
First, since $ G = \left (\begin {array} {cc}
I_ {k} & G_ {k} \\
\end {array} \right) $ is a matrix of $ k $ rows and $ 2k $ columns whose rows are linearly independent vectors, so the  application $ \psi $ is injective from $ \mathbb {F} _ {2} ^ {k} $ to $ \mathbb {F} _ {2} ^ {2k} $. Thus, $ \psi (\mathbb {F} _ {2} ^ {k}) $ is a vector space over $ \mathbb {F} _ {2} $ of dimension $ k $.
Then $ \psi $ forms a systematic linear error-correcting code of dimension $ k $ and length $ 2k $. \\
Now, let $ m $ be  the message such as $ m = m_ {1} \; m_ {2} ... m_ {k} $ and note by $ c $ its image by the application $ \psi $.
$$ c = \psi (m) = m \, \times \; G $$
Since $ \psi $ is a systematic code, a codeword $ c $ of length $ 2k $ can be distributed  into two vector $ c_ {1} $ and $ c_ {2} $. That is to say , $ c = \; c_ {1} \; c_ {2} $. Here, the vector $ c_1 $ is the original message ($ c_1 = m $) and $ c_2 $ is the vector (control bits) such that $ c_2 = m \; \times \; G_ {k} $. \\
So, for any integer $ i $ in $ \left \lbrace 1,2, ..., k \right \rbrace $, we have $$ c_ {2} [i] = \sum ^ {k} _ {j = 1, j \neq i} m_ {j}. $$
So, two cases are possible:
\begin{itemize}
\item[-] If the weight of $ m $ is even
\begin{equation*}
\left\lbrace \begin{array}{ccc}
\text{ and that } \;m_i = 0 & \Rightarrow & c_{2}[i]=\sum^{k}_{j=1,j\neq i} m_{j}=0.\\
\text{and if } \;m_i = 1 & \Rightarrow & c_{2}[i]=\sum^{k}_{j=1,j\neq i} m_{j}=1.
\end{array}\right. 
\end{equation*}
In this case, the code word $ c $ is: $ c = m \; m $. \\(Ex: for $ k = 6 $, if $ m = 011101 $ we have $ c = 011101 \; 011101 $ )
\item[-] If the weight of $ m $ is odd
\begin{equation*}
\left\lbrace \begin{array}{ccc}
\text{and that } \;m_i = 0 & \Rightarrow & c_{2}[i]=\sum^{k}_{j=1,j\neq i} m_{j}=1.\\
\text{and if }\; m_i = 1 & \Rightarrow & c_{2}[i]=\sum^{k}_{j=1,j\neq i} m_{j}=0.
\end{array}\right. 
\end{equation*}
In this case, the  code word $ c $ is: $ c = m \; \overline {m} $ where $ \overline {m} $ is the opposite of $ m $. \\(Ex: for $ k = 7 $, if $ m = 0000111 $ we have $ c = 0000111 \; 1111000 $)
\end{itemize}
Now, 
\begin{enumerate}
\item Take $k=3$,
\begin{equation}
 \begin{array}{ccc}
 \text{ if } w(m)=1 & \longrightarrow & w(c)= 3,\\
 \text{ if } w(m)=2 & \longrightarrow & w(c)= 4,\\
 \text{ if } w(m)=3 & \longrightarrow & w(c)= 3.\\
 \end{array}
 \end{equation}
Thus, the minimum distance for the code   $\mathcal{L}_3$  is 3.

\item for $k\geq 4$,
\begin{equation}
 \begin{array}{ccc}
 \text{ if } w(m)=1 & \longrightarrow & w(c)= k,\\
 \text{ if } w(m)=2 & \longrightarrow & w(c)= 4,\\
 \text{ if } w(m)=3 & \longrightarrow & w(c)= k,\\
  & \vdots  & \\
  \text{ if } w(m)=p \;(\text{even}) & \longrightarrow & w(c) = 2p,\\
   \text{if } w(m)=q \;(\text{odd}) & \longrightarrow & w(c) = k.\\
 \end{array}
 \end{equation}
Thus, the minimum distance for the code  $\mathcal{L}_k$ is 4.
\end{enumerate}
\end{proof}
\begin{example}
For $k=3$, from the main theorem (\ref * {thmpric1}), we have 
\begin{equation*}
 \mathcal{L}_3 =\left\lbrace \begin{array}{c}
 000000 \\
 001110 \\
 010101 \\
 100011 \\
 011011 \\
 101101 \\
 110110 \\
 111000 \\
  \end{array}\right\rbrace,\;\; G= \left(\begin{array}{cccccc}
    1 & 0 & 0 & 0 & 1 & 1\\
    0 & 1 & 0 & 1 & 0 & 1\\
    0 & 0 & 1 & 1 & 1 & 0\\
   \end{array}\right).\; \text{Initial messages:}
   \left\lbrace \begin{array}{c}
    000 \\
    001 \\
    010 \\
    100 \\
    011 \\
    101 \\
    110 \\
    111 \\
     \end{array}\right\rbrace.
\end{equation*}

\end{example}
\begin{example}
For $k=4$,from the main theorem (\ref*{thmpric1}), we have
 \begin{equation}\label{L4}
 G= \left(\begin{array}{cccccccc}
     1 & 0 & 0 & 0 & 0 & 1 & 1 & 1\\
     0 & 1 & 0 & 0 & 1 & 0 & 1 & 1\\
     0 & 0 & 1 & 0 & 1 & 1 & 0 & 1\\
     0 & 0 & 0 & 1 & 1 & 1 & 1 & 0\\
    \end{array}\right) 
    \text{ and } 
 \mathcal{L}_4 =\begin{array}{c}
 0000 0000 \\
 0001 1110 \\
 0010 1101 \\
 0100 1011 \\
 1000 0111 \\
  \end{array}
  \begin{array}{c}
   0011 0011 \\
   0101 0101 \\
   1001 1001 \\
   0110 0110 \\
   1010 1010 \\
   1100 1100 \\
    \end{array}
    \begin{array}{c}
     0111 1000 \\
     1011 0100 \\
     1101 0010 \\
     1110 0001 \\
     1111 1111
      \end{array}
\end{equation} 
  \end{example}
\begin{thm} \label{thmprinc2}(\textbf{The code  $\mathcal{L}_k ^{+}$} ).\\
Let $ k $ be an integer such that $ k \geq 4 $ and let's take the system $ \left \lbrace e '_ {1}, e' _ {2}, ..., e '_ {k} \right \rbrace $ where $ e '_ {i} = \psi (g_ {i}) $ (see the equation \ref {Gk+1}). So, we can build a $[3k,k]$-linear systematic code  with generator matrix 
 \begin{equation}
 G'= \left( \begin{array}{cc}
 & e'_{1}\\
 I_k & \vdots \\
  & e'_{k}
 \end{array}\right)
 \end{equation}
 witch is denoted by $\mathcal{L}_k ^{+}$.\\
The code $\mathcal{L}_4 ^{+}$ has minimum distance 5 and for $ k \geq 5 $, the code $\mathcal{L}_k ^{+}$ has minimum distance 6.
The generating matrix $ G $ of this code has the form
 \begin{equation}
  G'= \left( \begin{array}{ccc}
  I_k & G_k & I_k \\
  \end{array}\right).
 \end{equation}
\end{thm}
\begin{proof}
Since $\mathcal{L}_k $ is a sub-space over $ \mathbb {F} _2 $, any linear combination between the code words $ e' _ {1}, e '_ {2}, ..., e' _ {k} $ gives a code in $\mathcal{L}_k$ of weight equal to 4.\\ Then, for a message $ m $ in $ \mathbb {F} ^ {k} _2 $, the code word $ c $ generated by the matrix $ G '$ (ie $ c = m \; \times \; G '$) has a weight:
\begin{enumerate}
\item Pour $k=4$, 
\begin{equation}
 \begin{array}{ccc}
 \text{ if } w(m)=1 & \longrightarrow & w(c')= 5,\\
 \text{ if } w(m)=2 & \longrightarrow & w(c')= 6,\\
  \text{ if } w(m)=3 & \longrightarrow & w(c')\geq 7,\\
 & \vdots  & \\
 \text{ if } w(m)=k & \longrightarrow & w(c') \geq k+4.\\
 \end{array}
 \end{equation}
So we have a code 2-corrector $\mathcal{L}_4 ^{+}$ with a minimal distance $ d = 5 $.

 \item For $k\geq 5$,
 \begin{equation}
  \begin{array}{ccc}
  \text{ if } w(m)=1 & \longrightarrow & w(c')\geq 6,\\
  \text{ if } w(m)=2 & \longrightarrow & w(c') =6,\\
   \text{ if } w(m)=3 & \longrightarrow & w(c')\geq 7,\\
  & \vdots  & \\
  \text{ if } w(m)=k & \longrightarrow & w(c') \geq k+4.\\
  \end{array}
  \end{equation}
   So we have a code $ 2 $ -corrector $\mathcal{L}_k ^{+}$ with a minimal distance $ d = 6 $.
      
\end{enumerate}
\end{proof}

\begin{example} Le code $\mathcal{L}_4 ^{+}$. \\
Now take the four (4) vectors in $ \mathcal{L}_4$ which are: 
 \begin{equation}
 \begin{array}{c}
 e'_{1}=0111 1000 \\
 e'_{2}=1011 0100 \\
 e'_{3}=1101 0010 \\
 e'_{4}=1110 0001 \\
 \end{array}.
 \end{equation}
The code $\mathcal{L}_4 ^{+}$ is as follows:
  \begin{equation}
  G'=  \left( \begin{array}{cccccccccccc}
  \textcolor{green}{1} & 0 & 0 & 0 & \textcolor{green}{0} & 1 & 1 & 1 & \textcolor{green}{1} & 0 & 0 & 0 \\
  0 & \textcolor{green}{1} & 0 & 0 & 1 & \textcolor{green}{0} & 1 & 1 & 0 & \textcolor{green}{1} & 0 & 0 \\
  0 & 0 & \textcolor{green}{1} & 0 & 1 & 1 & \textcolor{green}{0} & 1 &  0 & 0 & \textcolor{green}{1} & 0  \\
  0 & 0 & 0 & \textcolor{green}{1} &  1 & 1 & 1 & \textcolor{green}{0} & 0 & 0 & 0 & \textcolor{green}{1} 
   \end{array}\right) . 
  \end{equation}
 \begin{equation*}
\mathcal{L}_4 ^{+}= \begin{array}{c}
  0000 \; 0000 \; 0000\\
  0001 \; 1110 \; 0001\\
  0010 \; 1101 \; 0010\\
  0100 \; 1011 \; 0100\\
  1000 \; 0111 \; 1000\\
  \end{array}
      \begin{array}{c}
     		 0011 \; 0011 \; 0011\\
             0101 \; 0101 \; 0101\\
             1001 \; 1001 \; 1001\\
             0110 \; 0110 \; 0110\\
             1010 \; 1010 \; 1010\\
             1100 \; 1100 \; 1100\\
      \end{array}
         \begin{array}{c}
          0111 \; 1000 \; 0111\\
          1011 \; 0100 \; 1011\\
          1101 \; 0010 \; 1101\\
          1110 \; 0001 \; 1110\\
          1111 \; 1111 \; 1111
         \end{array}
 \end{equation*}
\end{example}

\begin{example} The code $\mathcal{L}_5 ^{+}$. \\
Now take the five (5) vectors in $ \mathcal{L}_5$ which are:
 \begin{equation}
 \begin{array}{c}
 e'_{1}=01111 \; 10000 \\
 e'_{2}=10111 \; 01000 \\
 e'_{3}=11011 \; 00100 \\
 e'_{4}=11101 \; 00010 \\
 e'_{4}=11110 \; 00001 \\
 \end{array}.
 \end{equation}
The code $\mathcal{L}_5 ^{+}$ is thus as follows: \\
\begin{equation}
  G'= \left( \begin{array}{ccccccccccccccc}
  \textcolor{green}{1} & 0 & 0 & 0 & 0 & \textcolor{green}{0} & 1 & 1 & 1 & 1 &  \textcolor{green}{1} & 0 & 0 & 0 & 0\\
  0 & \textcolor{green}{1} & 0 & 0 & 0 & 1 & \textcolor{green}{0} & 1 & 1 & 1 & 0 & \textcolor{green}{1} & 0 & 0 & 0\\
  0 & 0 & \textcolor{green}{1} & 0 & 0 & 1 & 1 & \textcolor{green}{0} & 1 & 1 & 0 & 0 & \textcolor{green}{1} & 0 & 0\\
  0 & 0 & 0 & \textcolor{green}{1} & 0 & 1 & 1 & 1 & \textcolor{green}{0} & 1 & 0 & 0 & 0 & \textcolor{green}{1} & 0\\
  0 & 0 & 0 & 0 & \textcolor{green}{1} & 1 & 1 & 1 & 1 & \textcolor{green}{0} &  0 & 0 & 0 & 0 & \textcolor{green}{1}\\
   \end{array}\right) . 
  \end{equation}
The code words of $\mathcal{L}_5 ^{+}$ are:
 \begin{equation}\label{L5+}
\begin{array}{c}
 00000 \; 00000 \; 00000 \\
 00001 \; 11110 \; 00001 \\
 00010 \; 11101 \; 00010 \\
 00100 \; 11011 \; 00100 \\
 01000 \; 10111 \; 01000 \\
 10000 \; 01111 \; 10000 \\
  \end{array}
  \begin{array}{c}
   00011 \; 00011 \; 00011 \\
   00101 \; 00101 \; 00101 \\
   00110 \; 00110 \; 00110 \\
   01001 \; 01001 \; 01001 \\
   01010 \; 01010 \; 01010 \\
   01100 \; 01100 \; 01100 \\
   10001 \; 10001 \; 10001 \\
   10010 \; 10010 \; 10010 \\
   10100 \; 10100 \; 10100 \\
   11000 \; 11000 \; 11000 \\
  \end{array}
   \begin{array}{c}
     11100 \; 00011 \; 11100 \\
     11010 \; 00101 \; 11010 \\
     11001 \; 00110 \; 11001\\
     10110 \; 01001 \; 10110\\
     10101 \; 01010 \; 10101\\
     10011 \; 01100 \; 10011\\
     01110 \; 10001 \; 01110\\
     01101 \; 10010 \; 01101\\
     01011 \; 10100 \; 01011\\
     00111 \; 11000 \; 00111\\
    \end{array}
    \begin{array}{c}
     01111 \; 01111 \; 01111 \\
     10111 \; 10111 \; 10111\\
     11011 \; 11011 \; 11011 \\
     11101 \; 11101 \; 11101 \\
     11110 \; 11110 \; 11110 \\
     11111 \; 00000 \; 11111\\
      \end{array}
\end{equation}
\end{example}

\section{Study of the Weight of the codes $\mathcal{L}_k$ and $\mathcal{L}_k ^{+}$ }
Denote respectively by $\mathcal{L}_k|_{w=p}$ and $\mathcal{L}_k^{+}|_{w=p}$  the set of all elements of the code $\mathcal{L}_k$ and $\mathcal{L}_k ^{+}$ who have a weight equal to $p$.
\begin{thm}
For $\mathcal{L}_k$, the set of the weight of the words of this code form the set 
\begin{equation}
W(\mathcal{L}_k) = \left\lbrace 4i, i\in \mathbb{N} \mid 4i\leq 2k \right\rbrace \cup \left\lbrace  k\right\rbrace.
\end{equation}
\begin{itemize}
\item[-] For all weight in $\left\lbrace 4i, i\in \mathbb{N} \mid 4i\leq 2k \right\rbrace\backslash \left\lbrace  k\right\rbrace$ ,  $$\text{Card} \; \mathcal{L}_k|_{w=4i}=\left( \begin{array}{c}
k\\
2i
\end{array}\right)$$ 
\item[-] For the integer $k$ such that  $k \notin  \left\lbrace 4i, i\in \mathbb{N} \mid 4i\leq 2k \right\rbrace $, The set $\mathcal{L}_k|_{w=k}$ have cardinality, $$\text{Card} \; \mathcal{L}_k|_{w=k}=\sum_{i,\; 2i+1\leq k}\left( \begin{array}{c}
k\\
2i+1
\end{array}\right)$$
\item[-] For the integer $k$ such that $k\in  \left\lbrace 4i, i\in \mathbb{N} \mid 4i\leq 2k \right\rbrace $, the set $\mathcal{L}_k|_{w=k}$ have cardinality 
$$\text{Card} \; \mathcal{L}_k|_{w=k}=\sum_{i\mid \; 2i+1\leq k}\left( \begin{array}{c}
k\\
2i+1
\end{array}\right)+ \left( \begin{array}{c}
k\\
k/2
\end{array}\right) $$

\end{itemize}

\end{thm}

\begin{proof}
As defined above, a received code word $c$ in $\mathcal{L}_{k}$ can be split into two word, ie $ c = m_ {1} \; m_ {2} $ where $ m_ {1} $ is the message sent and $ m_ {2} $ which is the control code. Furthermore
$$m_2 = \left\lbrace \begin{array}{cc}
m_1 &  \text{if the weight of }   m_1 \text{ is even } \\
\overline{m_1} & \text{ if the weight of } m_1 \text{ is odd }
\end{array}\right. $$
So for a $m_1$ such that $w(m_1)=2i$, we have $w(c)=4i$ and  $$\text{ Card } \mathcal{L}_k|_{w=4i}=\text{card}\left\lbrace m_1\in \mathbb{F}_{2}^{k} \mid w(m_1)=2i\right\rbrace =  \left( \begin{array}{c}
k\\
2i
\end{array}\right).$$
Furthermore, for a $m_1$ such that $w(m_1)=2i+1$, we have $w(c)=k$  because $w(m_1)+w(m_2)=k$. then 
$$\text{ Card }\mathcal{L}_k|_{w=k}=\sum_{i\mid 2i+1\leq k} \text{card}\left\lbrace m_1\in \mathbb{F}_{2}^{k} \mid w(m_1)=2i+1\right\rbrace = \sum_{i\mid 2i+1\leq k} \left( \begin{array}{c}
k\\
2i+1
\end{array}\right).$$
For the integer $k$ such that $k\in  \left\lbrace 4i, i\in \mathbb{N} \mid 4i\leq 2k \right\rbrace $, we have, 
$$\mathcal{L}_k|_{w=k}=\cup_{i\mid 2i+1\leq k}\left\lbrace m_1\in \mathbb{F}_{2}^{k} \mid w(m_1)=2i+1\right\rbrace \cup \left\lbrace m_1\in \mathbb{F}_{2}^{k} \mid w(m_1)=4i, \;k=4i\right\rbrace  $$
So, the set $\mathcal{L}_k|_{w=k}$ have cardinality 
$$\text{Card} \; \mathcal{L}_k|_{w=k}= \left( \begin{array}{c}
k\\
k/2
\end{array}\right)+\sum_{i\mid \; 2i+1\leq k}\left( \begin{array}{c}
k\\
2i+1
\end{array}\right) $$
\end{proof}

\begin{example}
Take the code $\mathcal{L}_4$, The set $W(\mathcal{L}_4)$ is  $$ W(\mathcal{L}_4)=\left\lbrace 4i, i\in \mathbb{N} \mid 4i\leq 8 \right\rbrace \cup \left\lbrace  4\right\rbrace.$$ So,
$W(\mathcal{L}_4)=\left\lbrace 0,\; 4,\;8 \right\rbrace $. We have
 $$\text{Card}\;\mathcal{L}_4|_{w=0}=\left(\begin{array}{c}
 4\\
 0
 \end{array}\right) = 1 $$
As $4\in  \left\lbrace 4i, i\in \mathbb{N} \mid 2i\leq 4 \right\rbrace $, we have 
$$\text{Card} \; \mathcal{L}_4|_{w=4}= \left( \begin{array}{c}
4\\
2
\end{array}\right)+\left( \begin{array}{c}
4\\
1
\end{array}\right)+\left( \begin{array}{c}
4\\
3
\end{array}\right) = 14   $$

Finally,
$$\text{Card} \; \mathcal{L}_4|_{w=8}= \left( \begin{array}{c}
4\\
4
\end{array}\right) = 1 $$ 
Refer to \eqref{L4} to see the code $\mathcal{L}_4$.

\end{example}

\begin{thm}
For $\mathcal{L}_k^{+}$, the set of the weight of the words of this code form the set 
\begin{equation}
W(\mathcal{L}_k^{+}) = \left\lbrace 6i, i\in \mathbb{N} \mid 6i\leq 3k \right\rbrace \cup \left\lbrace  k+(2i+1), i \in \mathbb{N}\mid 2i+1\leq k\right\rbrace.
\end{equation}

\begin{itemize}
\item[-] For all weight in $\left\lbrace 6i, i\in \mathbb{N} \mid 6i\leq 3k \right\rbrace \backslash \left\lbrace  k+(2i+1), i \in \mathbb{N}\mid 2i+1\leq k\right\rbrace$   ,  $$\text{Card} \; \mathcal{L}_k|_{w=4i}=\left( \begin{array}{c}
k\\
2i
\end{array}\right)$$ 
\item[-] For all weight in $\left\lbrace  k+(2i+1), i \in \mathbb{N}\mid 2i+1\leq k\right\rbrace \backslash \left\lbrace 6i, i\in \mathbb{N} \mid 6i\leq 3k \right\rbrace $,\\ the set  $\mathcal{L}_k|_{w=k+(2i+1)}$ have cardinality, $$\text{Card} \; \mathcal{L}_k|_{w=k+(2i+1)}=\left( \begin{array}{c}
k\\
2i+1
\end{array}\right)$$
\item[-] For all weight $p$ in $\left\lbrace 6i, i\in \mathbb{N} \mid 6i\leq 3k \right\rbrace  \cap \left\lbrace  k+(2j+1), j \in \mathbb{N}\mid 2j+1\leq k\right\rbrace  $, the set $\mathcal{L}_k|_{w=p}$ have cardinality 
$$\text{Card} \; \mathcal{L}_k|_{w=p}=\left( \begin{array}{c}
k\\
2i
\end{array}\right)+ \left( \begin{array}{c}
k\\
2j+1
\end{array}\right) $$ where $p=6i=k+2j+1$.

\end{itemize}

\end{thm}
\begin{proof}
We know that for a  code $ c $ in $\mathcal{L}_k ^{+}$, this word have the form $ c = m_1 \; m_2 \; m_3 $ where $ m_1 = m_3 $ and that
\begin{equation*}
m_2=\left\lbrace \begin{array}{cc}
m_1 & \text{ if the weight of }  m_1 \text{ is even } \\
\text{or} & \\
\overline{m_1} & \text{ if the weight of } m_1 \text{is odd}
\end{array}\right.  
\end{equation*}
So for a $m_1$ such that $w(m_1)=2i$, we have $w(c)=6i$ and  $$\mathcal{L}_k|_{w=6i}=\text{card}\left\lbrace m_1\in \mathbb{F}_{2}^{k} \mid w(m_1)=2i\right\rbrace =  \left( \begin{array}{c}
k\\
2i
\end{array}\right).$$
Furthermore, for a $m_1$ such that $w(m_1)=2i+1$, we have $w(c)=k+(2i+1)$  because $w(m_1)+w(m_2)=k$. then 
$$\mathcal{L}_k|_{w=k+(2i+1)}=\text{card}\left\lbrace m_1\in \mathbb{F}_{2}^{k} \mid w(m_1)=2i+1\right\rbrace =  \left( \begin{array}{c}
k\\
2i+1
\end{array}\right).$$
For all weight $p$ in $\left\lbrace 6i, i\in \mathbb{N} \mid 6i\leq 3k \right\rbrace  \cap \left\lbrace  k+(2j+1), j \in \mathbb{N}\mid 2j+1\leq k\right\rbrace  $, we have, 
$$\mathcal{L}_k|_{w=p}=\left\lbrace m_1\in \mathbb{F}_{2}^{k} \mid w(m_1)=2i+1,\; p=k+(2i+1)\right\rbrace \cup \left\lbrace m_1\in \mathbb{F}_{2}^{k} \mid w(m_1)=2i, \;p=6i\right\rbrace  $$
So, the set $\mathcal{L}_k|_{w=k}$ have cardinality 
$$\text{Card} \; \mathcal{L}_k|_{w=p}=\left( \begin{array}{c}
k\\
2i
\end{array}\right)+ \left( \begin{array}{c}
k\\
2j+1
\end{array}\right) $$

\end{proof}

\begin{example}
For $\mathcal{L}_5^{+}$, the set $W(\mathcal{L}_5^{+})$ is 
\begin{equation*}
W(\mathcal{L}_5^{+}) = \left\lbrace 6i, i\in \mathbb{N} \mid 6i\leq 15 \right\rbrace \cup \left\lbrace  5+(2i+1), i \in \mathbb{N}\mid 2i+1\leq 5\right\rbrace.
\end{equation*}
So,
\begin{equation*}
W(\mathcal{L}_5^{+}) = \left\lbrace 0,\;6,\;12\right\rbrace \cup \left\lbrace 6,\;8\;10\right\rbrace = \left\lbrace 0,\;6,\;8,\;10,\;12\right\rbrace.
\end{equation*}
We have
 $$\text{Card}\;\mathcal{L}_5^{+}|_{w=0}=\left(\begin{array}{c}
 5\\
 0
 \end{array}\right) = 1. $$
 As $6\in \left\lbrace 0,\;6,\;12\right\rbrace \cap \left\lbrace 6,\;8\;10\right\rbrace$, we have 
  $$\text{Card}\;\mathcal{L}_5^{+}|_{w=6}=\left(\begin{array}{c}
  5\\
  2
  \end{array}\right) + \left(\begin{array}{c}
   5\\
   1
   \end{array}\right)= 15.$$
 Finally, 
  $$\text{Card}\;\mathcal{L}_5^{+}|_{w=8}=\left(\begin{array}{c}
  5\\
  3
  \end{array}\right) = 10, \text{ Card}\;\mathcal{L}_5^{+}|_{w=10}=\left(\begin{array}{c}
      5\\
      5
      \end{array}\right) = 1,$$
     $$\text{ and } \text{Card}\;\mathcal{L}_5^{+}|_{w=12}=\left(\begin{array}{c}
      5\\
      4
      \end{array}\right) = 5 $$
Refer to \eqref{L5+} to see the code $\mathcal{L}_5^{+}$.
      
\end{example}

\section{Decoding for the error correcting codes $\mathcal{L}_k$ and $\mathcal{L}_k ^{+}$ }
 After considering the parameters necessary for the study of these codes, we present here the appropriate decoding algorithms.
  \subsection{The dual codes of  $\mathcal{L}_k$ and $\mathcal{L}_k ^{+}$}
  \begin{thm}\label{Groebner} (See \cite{FJ,ME,JH} )\\
  If $C$ is an $[n, k ]$ code over $\mathbb{F}_{2}$, then the dual code $C^{\perp}$ is given by all words $u\in \mathbb{F}_{2} ^{n}$ such that  $<u\;,\;c>=0$ for each $c\in C$, where $<\;,\;>$ denotes the ordinary
  inner product. The dual code $C^{\perp}$ is an $[n, n-k]$ code. If $G=\left( I_{k}\ \mid M \right) $ is a generator matrix for $C$, then $H=\left( M^{T} \mid I_{n-k}\right) $ is the generator matrix for $C^{\perp}$.\\
   \end{thm}

   \begin{example}
  For the code $\mathcal{L}_{4}$ and $\mathcal{L}_{4}^{+}$ the generator matrix is respectively 
  \begin{equation*}
  G_{\mathcal{L}_{4}}= \left(\begin{array}{cccccccc}
       \textcolor{green}{1} & 0 & 0 & 0 & \textcolor{green}{0} & 1 & 1 & 1\\
       0 & \textcolor{green}{1} & 0 & 0 & 1 & \textcolor{green}{0} & 1 & 1\\
       0 & 0 & \textcolor{green}{1} & 0 & 1 & 1 & \textcolor{green}{0} & 1\\
       0 & 0 & 0 & \textcolor{green}{1} & 1 & 1 & 1 & \textcolor{green}{0}\\
      \end{array}\right)
    \end{equation*}
    and
    \begin{equation*}
           G_{\mathcal{L}_{4}^{+}}=  \left( \begin{array}{cccccccccccc}
         \textcolor{green}{1} & 0 & 0 & 0 & \textcolor{green}{0} & 1 & 1 & 1 & \textcolor{green}{1} & 0 & 0 & 0 \\
         0 & \textcolor{green}{1} & 0 & 0 & 1 & \textcolor{green}{0} & 1 & 1 & 0 & \textcolor{green}{1} & 0 & 0 \\
         0 & 0 & \textcolor{green}{1} & 0 & 1 & 1 & \textcolor{green}{0} & 1 &  0 & 0 & \textcolor{green}{1} & 0  \\
         0 & 0 & 0 & \textcolor{green}{1} &  1 & 1 & 1 & \textcolor{green}{0} & 0 & 0 & 0 & \textcolor{green}{1} 
          \end{array}\right) .
    \end{equation*}
  So the dual code $\mathcal{L}_{4}^{\perp}$ and $(\mathcal{L}_{4}^{+})^{\perp}$ have respectively his generator matrix:
   \begin{equation*}
  H_{\mathcal{L}_{4}^{\perp}}= \left(\begin{array}{cccccccc}
       \textcolor{green}{0} & 1 & 1 & 1 & \textcolor{green}{1} & 0 & 0 & 0 \\
       1 & \textcolor{green}{0} & 1 & 1 & 0 & \textcolor{green}{1} & 0 & 0 \\
       1 & 1 & \textcolor{green}{0} & 1 & 0 & 0 & \textcolor{green}{1} & 0 \\
       1 & 1 & 1 & \textcolor{green}{0} & 0 & 0 & 0 & \textcolor{green}{1} \\
       \end{array}\right)
     \end{equation*}
     and 
         \begin{equation*}
      H_{(\mathcal{L}_{4}^{+})^{\perp}}=  \left( \begin{array}{cccccccccccc}
           \textcolor{green}{0} & 1 & 1 & 1 & \textcolor{green}{1} & 0 & 0 & 0 & 0 & 0 & 0 & 0 \\
           1 & \textcolor{green}{0} & 1 & 1 & 0 & \textcolor{green}{1} & 0 & 0 & 0 & 0 & 0 & 0 \\
           1 & 1 & \textcolor{green}{0} & 1 & 0 & 0 & \textcolor{green}{1} & 0 & 0 & 0 & 0 & 0  \\
           1 & 1 & 1 & \textcolor{green}{0} & 0 & 0 & 0 & \textcolor{green}{1} & 0 & 0 & 0 & 0 \\
           \textcolor{green}{1} & 0 & 0 & 0 & 0 & 0 & 0 & 0 & \textcolor{green}{1} & 0 & 0 & 0 \\
           0 & \textcolor{green}{1} & 0 & 0 & 0 & 0 & 0 & 0 & 0 & \textcolor{green}{1} & 0 & 0 \\
           0 & 0 & \textcolor{green}{1} & 0 & 0 & 0 & 0 & 0 & 0 & 0 & \textcolor{green}{1} & 0 \\
           0 & 0 & 0 & \textcolor{green}{1} & 0 & 0 & 0 & 0 & 0 & 0 & 0 & \textcolor{green}{1}
            \end{array}\right) .
     \end{equation*}
   As the columns of $H_{\mathcal{L}_{4}^{\perp}}$ (or $H_{(\mathcal{L}_{4}^{+})^{\perp}}$)  are linearly independent, so for a codeword $c$ that contains exactly one error, the decoding will be easy by looking at the $H\times \;^{t}c$ syndrome.
  \end{example}
   \begin{defn}
   Let $c$ be an element of linear code $C$ such that $c=m_1\;m_2\;...\;m_n$ where $m_i \in \mathbb{F}_2$ for all $i$. we define the monomial $X^{c}$ of $\mathbb{F}_{2}[X_{1}X_{2}...X_{n}]$ by 
   \begin{equation}
   X^{c}= X_1^{m_{1}} X_2^{m_{2}}...X_n^{m_{n}} .
   \end{equation} 
   \begin{example}
   Take $c=101101$ which is a codeword of $\mathcal{L}_{3}$. In $\mathbb{F}_{2}[X_{1}X_{2}...X_{6}]$, The monomial $X^{c}$ was 
    $$X^{c}=X_{1}^{1}X_{2}^{0}X_{3}^{1}X_{4}^{1}X_{5}^{0}X_{6}^{1}=X_{1}X_{3}X_{4}X_{6}.$$
   \end{example}
   \end{defn}
   \begin{thm} (see \cite{MBQ})\\
   Let $C$ be an $[n,k]$-linear systematic code  over $\mathbb{F}_{2}$. Define the binomial ideal of $\mathbb{F}_{2}[X_{1}X_{2}...X_{n}]$ associated with $C$ the set $I_{C}$ such that 
   \begin{equation}
   I_{C}=\langle X^{c}-X^{c'} \mid c-c'\in C \rangle + \langle X_{i}^{2} - 1 \mid 1\leq i \leq n\rangle.
   \end{equation}
   \end{thm}
   
   \begin{thm} [Groebner basis of the Binomial ideal $I_C$](See \cite{ME})\\
   Take the lexicographic order on $\mathbb{F}_{2}[X_{1}X_{2}...X_{n}]$, i.e. $X_{1}\succ. . . \succ X_{n}$.
   An $[n,k]$ linear code systematic $C$ of generator matrix $(I_k \mid M )$ has the reduced Groebner basis 
   \begin{equation}\label{decodage groeb}
   \mathcal{B}=\left\lbrace  X_{i}-X^{m_i} \;\mid \; 1\leq i\leq k \right\rbrace \cup \left\lbrace X_{i} ^{2}-1 \;\mid\;  k+1\leq i\leq n\right\rbrace .
   \end{equation}
   Here $m_i$ is the $i^{th}$ line of the matrix $M$.
   \end{thm}
  \begin{example}
  For the code $(\mathcal{L}_{4}^{+})$, the corresponding binomial ideal of this code in $\mathbb{F}_{2}[X_{1}X_{2}...X_{n}]$  has the reduced Groebner basis given by the elements\\
  \begin{equation}
  \begin{array}{ccc}
  b_{1}= X_{1}- X_{6}X_{7}X_{8}X_{9} &  & b_{5}= X_{5}^{2}-1 \\
  b_{2}= X_{2}- X_{5}X_{7}X_{8}X_{10} &  & b_{6}= X_{6}^{2}-1 \\
  b_{3}= X_{3}- X_{5}X_{6}X_{8}X_{11} &  & b_{7}= X_{7}^{2}-1 \\
  b_{4}= X_{4}- X_{5}X_{6}X_{7}X_{12} &  & b_{8}= X_{8}^{2}-1 \\
  b_{9}= X_{9}^{2}-1  &  & b_{10}= X_{10}^{2}-1 \\
  b_{11}= X_{11}^{2}-1     &  & b_{12}= X_{12}^{2}-1\\
  \end{array}
  \end{equation}
  where 
   \begin{equation}
    G_{\mathcal{L}_{4}^{+}}=  \left( \begin{array}{cccccccccccc}
    \textcolor{green}{1} & 0 & 0 & 0 & \textcolor{green}{0} & 1 & 1 & 1 & \textcolor{green}{1} & 0 & 0 & 0 \\
    0 & \textcolor{green}{1} & 0 & 0 & 1 & \textcolor{green}{0} & 1 & 1 & 0 & \textcolor{green}{1} & 0 & 0 \\
    0 & 0 & \textcolor{green}{1} & 0 & 1 & 1 & \textcolor{green}{0} & 1 &  0 & 0 & \textcolor{green}{1} & 0  \\
    0 & 0 & 0 & \textcolor{green}{1} &  1 & 1 & 1 & \textcolor{green}{0} & 0 & 0 & 0 & \textcolor{green}{1} 
     \end{array}\right) . 
    \end{equation}
  \end{example}
  
  An immediate consequence of the study of the reduced Groebner basis of the Binomial ideal $I_C$ is a decoding algorithm for linear codes . This algorithm was given and proved in \cite{MBQ}.
  
 \subsection{Other way for error-correction of the code $\mathcal{L}_k$}
\begin{enumerate}
\item The ideal case is that no error was produced during transmission. We can use two methods to detect the presence of errors:
\begin{enumerate}
\item[-] We make the product of the control matrix $ H $  with the received code and we have to find a null vector, which means that there was no error during the transmission.
\item[-] Now the second method: as our code $\mathcal{L}_k$ is a systematic code, the received code word can be split into two, ie $ c = m_ {1} \; m_ {2} $ where $ m_ {1} $ is the message sent and $ m_ {2} $ which is the control code.\\
So, if the weight of $ m_ {1} $ is even and $ m_ {1} = m_ {2} $ or if the weight of $ m_ {1} $ is odd and $ m_ {2} = \overline {m_ {1}} $, in both cases the code has no error during the transmission. Otherwise there are errors.
\end{enumerate}

\item The other case is that errors occur during transmission.\\
Suppose that an error was produced. so,  we find out here   how to fix it. The only error must be in $ m_ {1} $ or $ m_ {2} $. Moreover, if the real message $ m $ sent is of even (odd) weight, the word $ m_ {1} + m_ {2} $ is of weight 1 (resp ($k-1$)). As a result, the decoding is as follows:
\begin {itemize}
\item[-] If the weight of $ m_ {1} + m_ {2} $ is 1, it remains to find the only bit that differs $ m_ {1} $ and $ m_ {2} $ and fix it for the weight of $ m_ {1}$ to be even.
\item[-] If the weight of $ m_ {1} + m_ {2} $ is $k-1$, it remains to find the only bit for that $ m_ {2} = \overline {m_1} $ and fix it for the weight of $ m_ {1}$ to be even.

\end{itemize}  
\end{enumerate}

\subsection*{Algorithm of Decoding for  the code $\mathcal{L}_k$}
\begin{center}
\framebox[1.2\width]{
\begin{minipage}[C]{.8\textwidth}
\begin{itemize}
\item[] Input $c$  (\textit{received word}) 
\item[] Output $c'$ (\textit{corrected word}) 
\item[] Begin
\begin{itemize}
\item [] Determined $m_1$ and $m_2$ such that $c=m_1 m_2$;
\item [] Calculate $w_1=w(m_1)$, $w_2=w(m_2)$ and $w_{1,2}=w(m_1 + m_2)$;
\begin{itemize}
\item If $w_{1,2}=0$ or $w_{1,2}=k$, so $c'=c$;
\item If  $w_{1,2}=1$
\begin{itemize}
\item and if $w_1$ is even, so $c' = m_1\; m_1$;
\item and if $w_1$ is odd, so $c' = m_2\; m_2$;
\end{itemize}
\item If  $w_{1,2}=k-1$
\begin{itemize}
\item and if $w_1$ is odd, so $c' = m_1\; \overline{m_1}$;
\item and if $w_1$ is even, so $c' = \overline{m_2}\; m_2$;
\end{itemize}
\item Else print("The message contains more than one error, we can not correct them")
\end{itemize}
\end{itemize}
\item[] End
\end{itemize}
\end{minipage}
}
\end{center}
\vspace{1cm}
\subsection{Other way for error correction of the code $\mathcal{L}_k ^{+}$}
 We try to give here the correction steps for a codeword that contains at most 2 errors.
 We know that for a  code $ c $ in $ \mathcal {L}^{+}$, this word have the form $ c = m_1 \; m_2 \; m_3 $ where $ m_1 = m_3 $ and that
 \begin{equation}
 m_2=\left\lbrace \begin{array}{cc}
 m_1 & \text{ if the weight of }  m_1 \text{ is even } \\
 \text{or} & \\
 \overline{m_1} & \text{ if the weight of } m_1 \text{is odd}
 \end{array}\right.  
 \end{equation}
In this sub-section, The word $ c $ be  the received word and suppose that it's contains at most 2 errors. 

\begin{thm}
If we have
$m_1=m_2=m_3$ for a word  $ m_1 $ which have even weight or if $m_1=m_3$ and $ m_ {2} = \overline {m_ {1}} $ for a word  $ m_1 $ which have odd weight, So this code word does not contain an error.
\end{thm}
\begin{proof}
By the construction in the theorem \ref{thmprinc2}, the received word $c$ is an element of  $ \mathcal {L}^{+}$. If we suppose that $c$ was wrong (errors less or equal than 2 bits), we have another codeword $c'$ in $ \mathcal {L}^{+}$ such that $c'$ is the exact code who transform into $c$. It is impossible because $w(c+c')\leq 2$ (in contradiction into $d=6$). So this code word does not contain an error. 
\end{proof}
\subsection*{Algorithm}
\begin{center}
\framebox[1.1\width]{
\begin{minipage}[C]{.8\textwidth}
\begin{itemize}
\item[] Input $c$  (\textit{received word}) 
\item[] Output $c'$ (\textit{corrected word}) 
\item[] Begin
\begin{itemize}
\item  Determined $m_1$, $m_2$ and $m_3$ such that $c=m_1\; m_2\; m_3$;
\item  Calculate $w_1=w(m_1)$ and $w_{1,2}=w(m_1 + m_2)$;
\begin{itemize}
\item [] If $w_1$ is even and $m_1 = m_2 =m_3$, so $c'=c$;
\item []  If $w_1$ is odd, $m_1=m_3$  and $w_{1,2}=k$, so $c'=c$;
\end{itemize}
\item Print("The message does not contain an error");
\end{itemize}
\item[] End
\end{itemize}
\end{minipage}
}
\end{center}

\vspace{0.5cm}

\begin{thm} Suppose that the word $c$ contain exactly an error.
\begin{enumerate}
 \item If $ w (m_1 + m_3) = 1 $ and  $ m_2 = m_3 $ or if $ w (m_1 + m_3) = 1 $ and $ w (m_2 + m_3) = k $, So an error is in $m_{1}$.\\
 For a correction, change $c=m_3\;m_2\;m_3.$
 \item If $ m_1 = m_3 $ and $ w (m_1 + m_2) = 1 $ or if $ m_1 = m_3 $ and $ w (m_1 + m_2) = k-1 $, so the error is in $ m_2 $.\\
 For a correction of this word: if the weight of $ m_1 $ is even, change $c=m_1\;m_1\;m_3$, else $c=m_1\;\overline{m_1}\;m_3.$
 \item If $ w (m_1 + m_3) = 1$ and $ m_1 = m_2 $ or if $ w (m_1 + m_3) = 1 $ and $ w (m_1 + m_2) = k$, so the error is in $m_3 $.\\
 For a correction of this word, change $c=m_1\;m_2\;m_1$.
 \end{enumerate}
\end{thm}
\begin{proof}If an error was present during transmission, three cases are possible.
\begin{enumerate}
\item If the error was in $m_1$, the exact message is $m_3$ and  $ w (m_1 + m_3) = 1 $. So, $ c = m_3 \; m_2 \; m_3 $ where 
 \begin{equation*}
 m_2=\left\lbrace \begin{array}{cc}
 m_3 & \text{ if the weight of }  m_3 \text{ is even } \\
 \text{or} & \\
 \overline{m_3} & \text{ if the weight of } m_3 \;\text{is odd} \;\; (\text{here} \; w (m_3 + m_2) = k).
 \end{array}\right.  
 \end{equation*}
 \item If the error was in $m_2$, it means that the exact message is $m_1$ and $ m_1=m_3$. So   
  \begin{equation*}
  w (m_1 + m_2) =\left\lbrace \begin{array}{cc}
  1 & \text{ if the weight of }  m_1 \; \text{ is even } \\
  \text{or} & \\
  k-1 & \text{ if the weight of } m_1 \; \text{is odd}.
  \end{array}\right.  
  \end{equation*}
  For correction,  $ c = m_1 \; m'_2 \; m_3 $ where 
  \begin{equation*}
  m'_2=\left\lbrace \begin{array}{cc}
  m_1 & \text{ if the weight of }  m_1 \; \text{ is even } \\
  \text{or} & \\
  \overline{m_1} & \text{ if the weight of } m_1 \; \text{is odd}.
  \end{array}\right.  
  \end{equation*}
  \item If the error was in $m_3$, the exact message is $m_1$. So $ w (m_1 + m_3) = 1 $  and $ c = m_1 \; m_2 \; m_1 $ where 
   \begin{equation*}
   m_2=\left\lbrace \begin{array}{cc}
   m_1 & \text{ if the weight of }  m_1 \text{ is even } \\
   \text{or} & \\
   \overline{m_1} & \text{ if the weight of } m_1 \;\text{is odd}.
   \end{array}\right.  
   \end{equation*}
\end{enumerate}
\end{proof}

\subsection*{Algorithm}
\begin{center}
\framebox[1.2\width]{
\begin{minipage}[C]{.8\textwidth}
\begin{itemize}
\item[] Input $c$  (\textit{received word}) 
\item[] Output $c'$ (\textit{corrected word}) 
\item[] Begin
\begin{itemize}
\item  Determined $m_1$, $m_2$ and $m_3$ such that $c=m_1\; m_2\; m_3$;
\item  Calculate $w_1=w(m_1)$, $w_3=w(m_3)$, $w_{1,2}=w(m_1 + m_2)$, $w_{2,3}=w(m_2 + m_3)$  and $w_{1,3}=w(m_1 + m_3)$;
\begin{itemize}
\item If $w_{1,3}=1$ 
\begin{itemize}
\item and if $m_2=m_3$, so $c'=m_3\;m_2\;m_3$;
\item and if $w_{2,3}=k$, so $c'=m_3\;m_2\;m_3$;
\item and if $m_2=m_1$, so $c'=m_1\;m_2\;m_1$;
\item and if $w_{1,2}=k$, so $c'=m_1\;m_2\;m_1$;
\end{itemize}
\item If $m_1=m_3$ and $w_{1,2}=1$ or $m_1=m_3$ and $w_{1,2}=k$, so 
\begin{itemize}
\item if $w_1$ is even,  $c'=m_1\;m_1\;m_3$;
\item else  $c'=m_1\;\overline{m_1}\;m_3$;
\end{itemize}
\end{itemize}
\end{itemize}
\item[] End
\end{itemize}
\end{minipage}
}
\end{center}

\vspace{0.5cm}

 \begin{thm} Suppose that the word $c$ contain exactly two errors and the both are in the word $m_1$ or $m_2$ or $m_3$.
 \begin{enumerate}
  \item If $ w (m_1 + m_3) = 2 $ and  $ m_2 = m_3 $ or if $ w (m_1 + m_3) = 2 $ and $ w (m_2 + m_3) = k $, the two errors are exactly in $m_{1}$.\\
  For a correction, change $c=m_3\;m_2\;m_3.$
  \item If $ m_1 = m_3 $ and $ w (m_1 + m_2) = 2 $ or if $ m_1 = m_3 $ and $ w (m_1 + m_2) = k-2 $, the two errors are exactly in $m_{2}$.\\
  For a correction of this word: if the weight of $ m_1 $ is even, change $c=m_1\;m_1\;m_3$, else $c=m_1\;\overline{m_1}\;m_3.$
  \item If $ w (m_1 + m_3) = 2 $ and $ m_1 = m_2 $ or if  $ w (m_1 + m_3) = 2 $ and $ w (m_1 + m_2) = k$, the two errors are exactly in $m_{3}$.\\
  For a correction of this word, change $c=m_1\;m_2\;m_1$.
  \end{enumerate}
 \end{thm}
 
 \begin{proof}
 If two errors are presented during transmission, we have 
 \begin{enumerate}
 \item If the both are in $m_1$, the exact message is $m_3$ and $ w (m_1 + m_3) = 2 $. So $ c = m_3 \; m_2 \; m_3 $ where 
  \begin{equation*}
  m_2=\left\lbrace \begin{array}{cc}
  m_3 & \text{ if the weight of }  m_3 \text{ is even } \\
  \text{or} & \\
  \overline{m_3} & \text{ if the weight of } m_3 \;\text{is odd}.
  \end{array}\right.  
  \end{equation*}
  \item If the both errors are in $m_2$, it means that the exact message is $m_1$ and $ m_1=m_3$. Then   
   \begin{equation*}
   w (m_1 + m_2) =\left\lbrace \begin{array}{cc}
   2 & \text{ if the weight of }  m_1 \; \text{ is even } \\
   \text{or} & \\
   k-2 & \text{ if the weight of } m_1 \; \text{is odd}.
   \end{array}\right.  
   \end{equation*}
   For correction,  $ c = m_1 \; m'_2 \; m_3 $ where 
   \begin{equation*}
   m'_2=\left\lbrace \begin{array}{cc}
   m_1 & \text{ if the weight of }  m_1 \; \text{ is even } \\
   \text{or} & \\
   \overline{m_1} & \text{ if the weight of } m_1 \; \text{is odd}.
   \end{array}\right.  
   \end{equation*}
   \item If the errors are in $m_3$, the exact message is $m_1$. So $ w (m_1 + m_3) = 2 $  and $ c = m_1 \; m_2 \; m_1 $ where 
    \begin{equation*}
    m_2=\left\lbrace \begin{array}{cc}
    m_1 & \text{ if the weight of }  m_1 \text{ is even } \\
    \text{or} & \\
    \overline{m_1} & \text{ if the weight of } m_1 \;\text{is odd}.
    \end{array}\right.  
    \end{equation*}
  
 \end{enumerate}
 \end{proof}
 
 \subsection*{Algorithm}
 \begin{center}
 \framebox[1.2\width]{
 \begin{minipage}[C]{.8\textwidth}
 \begin{itemize}
 \item[] Input $c$  (\textit{received word}) 
 \item[] Output $c'$ (\textit{corrected word}) 
 \item[] Begin
 \begin{itemize}
 \item  Determined $m_1$, $m_2$ and $m_3$ such that $c=m_1\; m_2\; m_3$;
 \item  Calculate $w_1=w(m_1)$, $w_3=w(m_3)$, $w_{1,2}=w(m_1 + m_2)$, $w_{2,3}=w(m_2 + m_3)$  and $w_{1,3}=w(m_1 + m_3)$;
 \begin{itemize}
 \item If $w_{1,3}=2$ 
 \begin{itemize}
 \item and if $m_2=m_3$, so $c'=m_3\;m_2\;m_3$;
 \item and if $w_{2,3}=k$, so $c'=m_3\;m_2\;m_3$;
 \item and if $m_2=m_1$, so $c'=m_1\;m_2\;m_1$;
 \item and if $w_{1,2}=k$, so $c'=m_1\;m_2\;m_1$;
 \end{itemize}
 \item If $m_1=m_3$
 \begin{itemize}
 \item and if $w_{1,2}=2$, so $c'=m_1\;m_1\;m_3$;
 \item and if $w_{1,2}=k-2$, so $c'=m_1\;\overline{m_1}\;m_3$;
 \end{itemize}
 \end{itemize}
 \end{itemize}
 \item[] End
 \end{itemize}
 \end{minipage}
 }
 \end{center}
 \vspace{0.5cm} 
 \begin{thm} Suppose that the received code $c$ contains exactly two errors distributed in $m_i$ and $m_j$ ($i\neq j$).
 \begin{enumerate}
 \item If $ w (m_1 + m_3) = 1 $ and $ w (m_3 + m_2) = 1 $, or if $ w (m_1 + m_3) = 1 $ and $ w (m_3 + m_2) = k-1 $, the two errors are distributed : one in $ m_1 $ and the other in $ m_2 $. 
 \item If $ w (m_1 + m_3) = 1 $ and $ w (m_1 + m_2) = 1 $ or if $ w (m_1 + m_3) = 1 $ and $ w (m_1 + m_2) = k-1 $, the two errors are distributed : one in $ m_2 $ and the other in $ m_3 $.
 \item  If $ w (m_1 + m_3) = 2 $ 
 \begin{itemize}
 \item if $ w (m_1 + m_2) = 1 $ and  $ w (m_2 + m_3) = 1 $ (the real message is of even weight)
 \item if $ w (m_1 + m_2) = k-1 $ and $ w (m_2 + m_3) = k-1 $ (the real message is odd weight),
 \end{itemize}
Then, the two errors are distributed : one in $ m_1 $ and the other in $ m_3 $. the value $ w (m_1 + m_3) = 0$  means that the error positions in $ m_1 $ and $ m_3 $ are different.
 \item  If $ w (m_1 + m_3) = 0 $ 
  \begin{itemize}
  \item if $ w (m_1 + m_2) = 1 $ and  $ w (m_2 + m_3) = 1 $ (the original message is of even weight)
  \item if $ w (m_1 + m_2) = k-1 $ and $ w (m_2 + m_3) = k-1 $ (the original message is odd weight),
  \end{itemize}
 So the two errors are distributed : one in $ m_1 $ and the other in $ m_3 $. the value $ w (m_1 + m_3) = 0 $  means that the error positions in $ m_1 $ and $ m_3 $ are the same. 
 
 \end{enumerate}
 
 \end{thm}

\begin{proof}
Suppose that $c$ contains two errors.
\begin{enumerate}
\item If the errors are distributed : one in $ m_1 $ and the other in $ m_2 $, the exact message is $m_3$ and $ w(m_1+m_3)=1$. Moreover   
   \begin{equation*}
   w (m_2 + m_3) =\left\lbrace \begin{array}{cc}
   1 & \text{ if the weight of }  m_3 \; \text{ is even } \\
   \text{or} & \\
   k-1 & \text{ if the weight of } m_3 \; \text{is odd}.
   \end{array}\right.  
   \end{equation*}
   For correction,  $ c = m_3 \; m'_2 \; m_3 $ where 
   \begin{equation*}
   m'_2=\left\lbrace \begin{array}{cc}
   m_3 & \text{ if the weight of }  m_3 \; \text{ is even } \\
   \text{or} & \\
   \overline{m_3} & \text{ if the weight of } m_3 \; \text{is odd}.
   \end{array}\right.  
   \end{equation*}
\item If the errors are distributed : one in $ m_2 $ and the other in $ m_3 $, it means that the exact message is $m_1$ and $ w(m_1+m_3)=1$. Moreover   
   \begin{equation*}
   w (m_1 + m_2) =\left\lbrace \begin{array}{cc}
   1 & \text{ if the weight of }  m_1 \; \text{ is even } \\
   \text{or} & \\
   k-1 & \text{ if the weight of } m_1 \; \text{is odd}.
   \end{array}\right.  
   \end{equation*}
   For correction,  $ c = m_1 \; m'_2 \; m_1 $ where 
   \begin{equation*}
   m'_2=\left\lbrace \begin{array}{cc}
   m_1 & \text{ if the weight of }  m_1\; \text{ is even } \\
   \text{or} & \\
   \overline{m_1} & \text{ if the weight of } m_1 \; \text{is odd}.
   \end{array}\right.  
   \end{equation*} 
   \item Now suppose that the errors are distributed : one in $m_1$ and the other in $m_3$ and the error positions in $ m_1 $ and $ m_3 $ are different. So we have $w(m_1 +m_3)=2$. 
   \begin{itemize}
   \item If the original message is of even weight i.e. $w(m_1)$ is odd, so $w(m_1 +m_2)=1$ and $w(m_2 +m_3)=1$.\\
   For correction, we have $c=m_2\;m_2\;m_2$.
   \item If the original message is of odd weight i.e. $w(m_1)$ is even, so $w(m_1 +m_2)=k-1$ and $w(m_2 +m_3)=k-1$.\\
      For correction,we have $c=\overline{m_2}\;m_2\;\overline{m_2}$.
   \end{itemize}

	\item Suppose that the errors are distributed : one in $m_1$ and the other in $m_3$ and the error positions in $ m_1 $ and $ m_3 $ are the same. So we have $w(m_1 +m_3)=0$ i.e. $m_1=m_3$. 
      \begin{itemize}
      \item If the original message is of even weight i.e. $w(m_1)$ is odd, so $w(m_1 +m_2)=1$ and $w(m_2 +m_3)=1$.\\
      For correction,we have $c=m_2\;m_2\;m_2$.
      \item If the original message is of odd weight i.e. $w(m_1)$ is even, so $w(m_1 +m_2)=k-1$ and $w(m_2 +m_3)=k-1$.\\
         For correction,we have $c=\overline{m_2}\;m_2\;\overline{m_2}$.
      \end{itemize}
   \end{enumerate}
\end{proof}
 \subsection*{Algorithm}
 \begin{center}
 \framebox[1.2\width]{
 \begin{minipage}[C]{.8\textwidth}
 \begin{itemize}
 \item[] Input $c$  (\textit{received word}) 
 \item[] Output $c'$ (\textit{corrected word}) 
 \item[] Begin
 \begin{itemize}
 \item  Determined $m_1$, $m_2$ and $m_3$ such that $c=m_1\; m_2\; m_3$;
 \item  Calculate $w_1=w(m_1)$, $w_3=w(m_3)$, $w_{1,2}=w(m_1 + m_2)$, $w_{2,3}=w(m_2 + m_3)$  and $w_{1,3}=w(m_1 + m_3)$;
 \begin{itemize}
 \item  If $w_{1,3}=1$ 
 \begin{itemize}
 \item  and if $w_{2,3}=1$, so $c'=m_3\;m_3\;m_3$;
 \item  and if $w_{2,3}=k-1$, so $c'=m_3\;\overline{m_3}\;m_3$;
 \item  and if $w_{1,2}=1$, so $c'=m_1\;m_1\;m_1$;
 \item  and if $w_{1,2}=k-1$, so $c'=m_1\;\overline{m_1}\;m_1$;
 \end{itemize}
 \item  If $w_{1,3}=2$ 
  \begin{itemize}
  \item and if $w_{1,2}=1$ and $w_{2,3}=1$, so $c'=m_2\;m_2\;m_2$;
   \item and if $w_{1,2}=k-1$ and $w_{2,3}=k-1$, so $c'=\overline{m_2}\;m_2\;\overline{m_2}$;
    \end{itemize}
    \item  If $w_{1,3}=0$ 
      \begin{itemize}
      \item and if $w_{1,2}=1$ and $w_{2,3}=1$, so $c'=m_2\;m_2\;m_2$;
       \item and if $w_{1,2}=k-1$ and $w_{2,3}=k-1$, so $c'=\overline{m_2}\;m_2\;\overline{m_2}$;
        \end{itemize}
 \end{itemize}
 \end{itemize}
 \item[] End
 \end{itemize}
 \end{minipage}
 }
 \end{center}
 \vspace{1cm}

\subsection{Comparative analysis between $\mathcal{L}_k ^{+}$ and the code of Hamming}
The table below presents a comparative study between the Hamming code and our  $\mathcal{L}_k ^{+}$ for some cases where the two codes are of the same length.\\
\begin{equation*}
\begin{array}{c|c|c}
& \text{HAMMING CODE} & \text{CODE BUILT FROM }\\
&  & \text{SUB-EXCEEDING  FONCTION}\\\hline
 & & \\
 & \text{It is a linear code of  } & \text{It is a linear code of  }\\
 & \text{the form } & \text{the form } \\
 & [2^{r}-1,2^{r}-r-1,3]. & [3k,k,6]. \\
 & & \\
\text{Form} & \text{For } r=3,\; \text{we have} & \text{For } k=5,\; \text{we have}\\
 & [15,11,3]-\text{linear code}. & [15,5,6]-\text{linear code}.\\
  & & \\
 & \text{For } r=6,\; \text{we have} &  \text{For } k=21,\; \text{we have} \\
  & [63,57,3]-\text{linear code}. & [63,21,6]-\text{linear code}.\\
   & & \\\hline
   & & \\
   &\text{Minimum distance:}& \text{Minimum distance:}\\
    & d=3 & d=6\\
      & \text{Correction capacity:}& \text{Correction capacity:}\\
        & e_{c}=1 & e_{c}=2\\
           & & \\
   \text{Paremeters}   & \text{For the  } [15,11,3]-\text{code}& \text{For the  } [15,5,6]-\text{code}\\
    & \text{Correction rate is } & \text{Correction rate is } \\
    & C_{r}=1/15 & C_{r}=2/15 \\
    & & \\
  & \text{For the  } [63,57,3]-\text{code} & \text{For the  } [63,21,6]-\text{code}\\
    & \text{Correction rate is } &  \text{Correction rate is } \\
     & C_{r}=1/63 & C_{r}=2/63\\
      & & \\\hline
\end{array}
\end{equation*}
These codes have the same length but what make the differences are the dimensions. However, the $[3k,k,6]$-linear systematics code has 2 bit for the correction capacity and then the error detection capability was 5 against the Hamming code which only corrects an error and detects only 2 errors over the length $2^{r}-1$.
\section*{Acknowledgement}
The authors would like to thank the anonymous referees for their careful reading of our manuscript and
helpful suggestions.
\section*{Ethics}
The authors declare that there is no conflict interests regarding the publication of this manuscript.


\end{document}